\documentclass[12pt]{article}

\usepackage{comment}

\usepackage{multirow}
\usepackage{bm}
\usepackage{amssymb}
\usepackage{color}
\usepackage{mathrsfs}
\usepackage{mathtools}

\usepackage{natbib}
\usepackage{xr}

\bibpunct{(}{)}{,}{autheryear}{,}{;}
\newcommand{\blind}{1}
\newcommand{\bbP}{\mathbb{P}}

\newcommand{\bx}{\bm{x}}

\newcommand{\bb}{\bm{b}}
\newcommand{\bd}{\bm{d}}
\newcommand{\bp}{\bm{p}}
\newcommand{\by}{\bm{y}}

\newcommand{\bz}{\bm{z}}

\newcommand{\bX}{\mathbf{X}}
\newcommand{\bY}{\mathbf{Y}}
\newcommand{\bU}{\mathbf{U}}
\newcommand{\bV}{\mathbf{V}}

\newcommand{\bphi}{\boldsymbol{\phi}}

\newcommand{\bbeta}{\boldsymbol{\beta}}
\newcommand{\bdelta}{\boldsymbol{\delta}}

\newcommand{\btheta}{\boldsymbol{\theta}}
\newcommand{\btau}{\boldsymbol{\tau}}

\newcommand{\blambda}{\boldsymbol{\lambda}}

\newcommand{\balpha}{\boldsymbol{\alpha}}

\newcommand{\bzero}{\mathbf{0}}

\def\trans{^{\rm T}}

\usepackage{mathabx}

\def\wb{\widebar}
\def\wh{\widehat}
\def\wt{\widetilde}

\newtheorem{theorem}{Theorem}
\newtheorem{proposition}{Proposition}
\newtheorem{lemma}{Lemma}
\newtheorem{corollary}{Corollary}
\newtheorem{remark}{Remark}

\usepackage[linesnumbered,ruled]{algorithm2e}

\begin{document}

\def\spacingset#1{\renewcommand{\baselinestretch}%
{#1}\small\normalsize} \spacingset{1}

\if1\blind
{
  \title{\bf Information projection approach to  propensity score  estimation for handling selection bias under missing at random}
  \author{Hengfang Wang$^1$  and
     Jae Kwang Kim$^2$ \thanks{{jkim@iastate.edu}, Department of Statistics, Iowa State University, Ames, Iowa, U.S.A.} \\
    $^{1}$ School of Mathematics and Statistics, Fujian Normal University\\
    $^2$ Department of Statistics, Iowa State University}
\date{}
  \maketitle
} \fi

\if0\blind
{
  \bigskip
  \bigskip
  \bigskip
  \begin{center}
    {\LARGE\bf Information projection approach to  propensity score  estimation for handling selection bias under missing at random}
\end{center}
  \medskip
} \fi

\bigskip

\noindent%

\textbf{Summary.}
Propensity score weighting is widely used to improve the representativeness and correct the selection bias in the voluntary sample. The  propensity  score  is  often  developed using a model for the sampling probability, which can be subject to model misspecification.  In this paper, we consider an alternative approach of estimating the inverse of the propensity scores using the density ratio function satisfying the self-efficiency condition. The smoothed density ratio function is obtained by the solution to the information projection onto the space satisfying the moment conditions on the balancing scores. 
 By including the covariates for the outcome regression models only in  the density ratio model, we can achieve efficient propensity  score  estimation.  Penalized regression is  used to identify important covariates. 
We further extend the proposed approach to the multivariate missing case. Some limited simulation studies are presented to  compare with the existing methods.


\noindent%
\vfill


\newpage
\spacingset{1}

\section{Introduction}
In statistical analysis using sample data, the concept of representativeness is crucial. The fundamental issue is that the final sample might be subject to selection bias and might not accurately reflect the target population. A naive analysis that ignores selection bias can produce incorrect results. Another way to frame the selection bias problem is as a missing data problem. 
Making valid statistical
inference with missing data is a fundamental problem in statistics \citep{little2019statistical, kim2013statistical}. 

The propensity score (PS) weighting approach is  widely used  for adjusting the selection bias in the final sample. Frequently, the propensity score is calculated using the  selection  probability model. In principle, regression models for binary responses, such as logistic regression, can be used to estimate the selection   probability in the presence of observed auxiliary data. The target parameter can then be estimated using an inverse probability weighting estimator to provide an unbiased estimate. However, correctly specifying the propensity score model can be tricky, and we frequently lack sufficient understanding of the selection mechanism.  Furthermore, the final estimation can be unstable when some propensity scores are close to zero \citep{ma2020}. The bias and variance of the resulting estimator are likely to be amplified by the nature of such an  `inverse' fashion.

{ The existing methods for propensity score estimation are either based on the maximum likelihood method
\citep{rosenbaum1983central, robins1994estimation}
 or calibration method \citep{folsom1991exponential,tan2010bounded, kot10,  graham2012inverse,hainmueller2012entropy, imai2014covariate,  chan2016globally, zhao2019}. The calibration method can also be justified under an outcome model, which  gives a doubly robust flavor. However,  the objective function for calibration estimation is not fully agreed. } 

\textcolor{black}{In this paper, we introduce the self-efficiency condition as a new paradigm for PS estimation. The self-efficiency condition means that the resulting PS estimator is equivalent to the prediction-based  estimator using the outcome response model. Due to the equivalence, the self-efficient propensity score estimator is more efficient than existing approaches. To build a self-efficient PS estimation, we first model density ratios with information projection \citep{csiszar2004}, which minimizes the Kullback-Leibler divergence measure for models with moment constraints on the balance scores. Density ratio models are widely employed to address classification and two-sample problems \citep{qin1998, nguyen2010, chen2013}.  To our best knowledge, however, it has not been studied in the context of correcting selection bias. Once the density ratio model has been constructed using the information projection approach, the model parameters can be estimated using the empirical likelihood technique  \citep{qin1994}. }

\textcolor{black}{
The self-efficient PS estimation is also closely related to weight smoothing. Information projection of the density ratio function combined with  calibration  estimation is  used to accomplish weight smoothing. We give some asymptotic theories for the smoothed PS estimator utilizing the  balance score function. The proposed estimator is doubly robust and locally efficient in the sense of \cite{robins1994estimation}.  By applying the weight smoothing based on the covariate in the outcome model, we can achieve the efficiency gain. 
 In addition, a two-step estimation approach is designed to facilitate parameter estimation computation.   
Variance estimation of the PS estimator can be implemented using either the linearization method or bootstrap. 
}

Furthermore, the proposed paradigm includes discussing propensity score estimation in high-dimensional covariate scenarios. When there are more auxiliary variables than we require, the PS estimator can be inefficient. Only significant covariates for the outcome model should be included in the log density ratio model for effective PS estimation. 
We consider  penalized regression to identify important covariates. Given the observed study variable and the corresponding auxiliary variables, we can implement penalized  regression methods  to select important covariates and obtain an  efficient PS estimator. 

Finally, using the proposed framework, the PS estimation can be easily extended to handle multivariate missing data.  
We can partition the sample into multiple groups based on missing patterns and apply the density ratio estimation method to obtain the inverse propensity scores.  By constructing a self-efficient PS estimation, we can
combine information from multiple sources with different missing patterns. 

The paper is organized as follows. In Section 2, the basic setup is introduced. In Section 3, the proposed method is developed using the information projection technique. In Section 4, 
its  asymptotic properties are investigated. In Section 5, the proposed method is extended to handle high dimensional covariate cases. In Section 6, an extension to multivariate missing data setup is developed. Two limited simulation studies are presented in Section 7 to compare with other existing methods and to understand the effect of choice of calibration variables. 
Concluding remarks are made in Section 8.


\section{Basic Setup}

Suppose that 
the parameter of interest $\btheta \in \mathbb{R}^{p}$ can be written as the unique solution to $\mathbb{E}\{ \bU( \btheta; \bX, Y) \}=\bzero$, where $Y$ is the study variable that is subject to missingness, $\bX$ is the auxiliary variable that is always observed, and $\bU$ is a smooth function of $\btheta$ with nonsingular gradient matrix $\mathbb{E}\{ \partial_\theta \bU( \btheta) \}$. Thus, the joint density of $(\bX, Y)$ is completely unspecified except for the moment condition  $\mathbb{E}\{ \bU( \btheta; \bX, Y) \}=\bzero$. 
We further assume that the second moments of $\bU( \btheta; \bX, Y)$ exist.

 Suppose we have  an $N$ independently and identically distributed realization of $(\bX, Y)$, denoted as $\{(\bx_{i}, y_{i}): i=1, \ldots, N \}$.   Since $Y$ is subject to missingness,  the actual dataset  is $\{(\bx_{i}, \delta_{i}y_{i}, \delta_{i}): i=1, \cdots, N\}$, where $\delta_{i}$ is the sampling indicator variable defined as 
\begin{align}
\delta_{i} = \begin{cases}
1 & \mbox{if\ } y_{i} \mbox{\ is observed,}\\
0 & \mbox{otherwise.}\\
\end{cases}
\notag
\end{align}

\textcolor{black}{
We  assume that the sampling  mechanism is 
 missing at random (MAR)   in the sense of \cite{rubin1976}. That is,  
\begin{align}
Y \perp \delta \mid \bX. 
\label{mar} 
\end{align} 
Under MAR, we are interested in finding 
 the propensity weights  $\omega(\bx_i)$ such that the solution to 
\begin{equation} 
 \sum_{i=1}^N \delta_i \omega(\bx_i) U( \btheta; \bx_i, y_i) = 0 
\label{ps}
\end{equation} 
leads to consistent estimation of $\btheta$. Equation (\ref{ps}) is called the propensity score equation. 
One popular choice for $\omega(\bx)$ is 
\begin{equation} 
\omega(\bx) =\{\mathbb{P} ( \delta =1 \mid \bx) \}^{-1}.
\label{direct}
\end{equation} 
There are two problems with the choice in (\ref{direct}). First, we do not know the true propensity score function $\mathbb{P} ( \delta =1 \mid \bx)$. 
Second, the resulting estimator is not necessarily efficient. To resolve these problems, we may consider the following conditions: 
\begin{equation} 
\sum_{i=1}^N \delta_i \omega(\bx_i) U( \btheta; \bx_i, y_i) =  \sum_{i=1}^N \left[ \delta_iU( \btheta; \bx_i, y_i) + (1- \delta_i) \hat{\mathbb{E}}\{ U( \btheta; \bx_i, Y) \mid \bx_i\}  \right],
\label{self} 
\end{equation} 
where $\hat{\mathbb{E}}\{ U( \btheta; \bx_i, Y) \mid \bx_i\}$ is the best predictor of $U( \btheta ; \bx_i, y_i)$ for $\delta_i=0$. 
Condition (\ref{self}) is attractive as the solution to the PS estimating equation in (\ref{ps}) is  equivalent to the expected estimating equation, which leverages the best predictors of the unobserved parts of the population estimating equation. We call condition (\ref{self}) \emph{self-efficiency} condition, as we can obtain an efficient estimator of $\btheta$ by solving the expected estimating equation \citep{wang2000}. The term ``self'' is employed to emphasize that efficiency is achieved without including additional terms to the propensity score equation.
}

\textcolor{black}{
To compute the conditional expectation in (\ref{self}), one could  make an extra assumption on the outcome model $[ y \mid \bx]$ and estimate the parameters in the model. To avoid the difficulty associated with the modeling for $[y \mid \bx]$, we  assume that 
\begin{equation} 
 \mathbb{E}\{ U( \btheta; \bx, Y) \mid \bx\}  \in \mbox{span} \{ 1,  b_1( \bx), \ldots, b_L( \bx) \} := \mathcal{H}
\label{basis}
\end{equation} 
for some $b_1( \bx), \ldots, b_L( \bx)$. Once condition  (\ref{basis}) is met, the additional modeling for $[y \mid \bx]$ can be omitted.  
Under  (\ref{basis}), by taking the conditional expectations  on both terms in (\ref{self}), we can check that the self-efficiency condition implies 
 \begin{equation} 
    \sum_{i=1}^N \delta_i \omega ( \bx_i) \bb (\bx_i) =  \sum_{i=1}^N  \bb (\bx_i) , 
   \label{balance} 
   \end{equation} 
   where $\bb( \bx)=(1, b_1(\bx), \cdots$ $, b_L(\bx))$ is a vector of  basis  functions in $\mathcal{H}$. Thus, the self-efficient propensity score method achieves the balancing property for $\bb(\bx)$. The balancing property is also called calibration property in survey sampling \citep{deville1992, fuller2009, 
   wu2020}.  
}

Another implication of the self-efficiency is related to weight smoothing.  
Note that, under assumption (\ref{basis}), the self-efficiency condition in (\ref{self}) implies that $\omega(\bx_i)$ is a function of $b_1(\bx), \cdots, b_L( \bx)$, which means that the final PS weights are functions of the covariates in the outcome regression model. By including the covariates in the outcome model only, we can achieve weight smoothing and efficient estimation.

\textcolor{black}{Thus, under assumption (\ref{basis}), it makes sense to impose the self-efficiency  for obtaining the propensity weights. 
How to uniquely determine $\omega(\bx)$ to satisfy the self-efficiency condition is our main research problem. We address this problem by introducing the density ratio function and information projection in the next section. 
}

\section{Proposed method }\label{EM_DRE}

\subsection{Density ratio function} 

To introduce the density ratio function, we use $f(\cdot)$ to denote generic density functions. 
Further, let $f_{j}(\bx)$ denote the conditional density of $\bX$ given $\delta = j$ for $j = 0, 1$. Using this notation, we define the following density ratio function: 
\begin{align*}
{r}(\bx) = \frac{f_{0}(\bx)}{ f_{1}(\bx)}.
\end{align*}
By Bayes theorem, we obtain 
$$ \frac{ \mathbb{P} ( \delta =0 \mid \bx) }{\mathbb{P} ( \delta =1 \mid \bx)  } = \frac{\mathbb{P}( \delta=0)}{ \mathbb{P}( \delta = 1) } \times r(\bx). 
$$
Assuming $c= \mathbb{P}( \delta=0)/ \mathbb{P}( \delta = 1)$ is known, we can express 
\begin{equation} 
 \frac{1}{ \mathbb{P} ( \delta =1 \mid \bX) } = 1 + c \cdot  r(\bX) 
\label{1}
\end{equation} 
and use 
\begin{equation*} 
 \omega_i = \frac{1}{ \mathbb{P} ( \delta_i =1 \mid \bx_i) } = 1 + c \cdot  r(\bx_i) 
 \end{equation*} 
as the propensity score weight for unit $i$ with $\delta_i=1$. So, our propensity score weighting problem reduces to density ratio function  estimation problem. Note that we can easily estimate $c$ by $\wh{c} = N_0/N_1$, $N_1= \sum_{i=1}^N \delta_i$ and $N_0 = N-N_1$. 




\textcolor{black}{
Let $\bb( \bx)=( b_1(\bx), \cdots$ $, b_L(\bx))$ be a vector of basis functions of $\mathcal{H}$. Condition  (\ref{basis}) implies that, as far as estimation of $\btheta$ is concerned, the MAR condition holds conditional on $\bb( \bX)$, that is, 
 \begin{equation} 
 Y \perp \delta \mid \bb( \bX) . 
 \label{mar2} 
 \end{equation} 
 According to \cite{rosenbaum1983central},   $\bb( \bx)$ in (\ref{mar2}) is called the balancing score function.
}

Using $\bb(\bx)$ satisfying (\ref{mar2}), we can construct \begin{equation} 
 r^\star (\bx) = \mathbb{E}\left\{ r(\bx) \mid \bb(\bx)  , \delta = 1 \right\} ,
 \label{bdr}
 \end{equation} 
 where $r(\bx)=f_0(\bx)/f_1(\bx)$ is the true density ratio function. 
  To distinguish ${r}^{\star}(\bx)$ from the original density ratio function $r(\bx)$, we may call ${r}^{\star}(\bx)$  the smoothed density ratio function in the sense that  it satisfies 
  \begin{equation} 
   \frac{1}{N_1} \sum_{i=1}^N \delta_i r^{\star} ( \bx_i) [1, \bb (\bx_i)] = \frac{1}{N_0} \sum_{i=1}^N \left( 1- \delta_i \right) [1, \bb (\bx_i)]  .
   \label{balance} 
   \end{equation} 
   Thus, the smoothed density ratio function $r^{\star} (\bx)$ is regarded as a low-dimensional projection of $r( \bx)$ onto the space satisfying (\ref{balance}).

 We can use ${r}^{\star} (\bx)$  in (\ref{bdr}) to construct a smoothed PS estimating function of $\btheta$ as follows:  
\begin{equation} 
 \wh{\bU}_{\rm SPS} ( \btheta ) = \frac{1}{N} \sum_{i=1}^N \delta_i \left\{ 1+ (N_0/N_1) r^{\star}( \bx_i) \right\} \bU \left( \btheta; \bx_i, y_i \right).
 \label{ps2} 
 \end{equation} 
 Furthermore, since $r(\bx)$ is a function of $\bx$ only,  we obtain 
 \begin{equation} 
  \mathbb{E}\left\{ r(\bx) \mid \bb(\bx) , y, \delta = 1 \right\} = \mathbb{E}\left\{ r(\bx) \mid \bb(\bx) , \delta = 1 \right\}, 
 \label{result2}
 \end{equation} 
 where $r(\bx)=f_0(\bx)/f_1(\bx)$ is the true density ratio function. 
 Note that result (\ref{result2}) implies that 
 $$ \mathbb{E}\{ \wh{\bU}_{\rm PS} (\btheta) \mid \bb(\bx), y, \bdelta \} = \wh{\bU}_{\rm SPS} (\btheta) , $$
where 
$$\wh{\bU}_{\rm PS} (\btheta)
= \frac{1}{N} \sum_{i=1}^N  \delta_i \left\{ 1+ (N_0/N_1) {r}( \bx_i) \right\} \bU \left( \btheta; \bx_i, y_i \right) 
. $$ 
Therefore, we can establish the following result. 
\begin{proposition}\label{P2}
Under (\ref{mar2}), we obtain 
\begin{equation} 
\mathbb{E} \left\{ \wh{\bU}_{\rm PS} ( \btheta ) \right\} = \mathbb{E}   \left\{ \wh{\bU}_{\rm SPS} ( \btheta ) \right\} 
\end{equation} 
and 
\begin{equation} 
\mathbb{V} \left\{ \wh{\bU}_{\rm PS} ( \btheta ) \right\}  \ge \mathbb{V}  \left\{ \wh{\bU}_{\rm SPS} ( \btheta ) \right\}   .
\label{result3} 
\end{equation}
\label{Pro1}
\end{proposition} 

Proposition \ref{P2} implies that the solution to $\wh{\bU}_{\rm SPS} (\btheta)=\bzero$ is more efficient than the solution to $\wh{\bU}_{\rm PS} (\btheta)=\bzero$. Such phenomenon has been recognized by \cite{little05},  \cite{beaumont2008}, \cite{kim2013}  and \cite{park2019b}. 
To find a model of the smoothed density ratio function $r^{\star}( \bx)$ in (\ref{bdr}), we introduce  information projection in the following subsection.

\subsection{Information projection} 


To motivate the proposed method, suppose that we have two probability distributions $\mathbb{P}_0, \mathbb{P}_1$, with $\mathbb{P}_0$ absolutely continuous with respect to $\mathbb{P}_1$. For simplicity, we also assume that $\mathbb{P}_k$ are absolutely continuous with respect to Lebesgue measure $\mu$, with density $f_k$ with support $\mathcal{X} \subset \mathbb{R}^p$, for $k=0,1$. The Kullback-Leibler  divergence between $\bbP_0$ and $\bbP_1$ is defined by 
$$D( \bbP_0\parallel \bbP_1) = \int \log \left(\frac{ d \bbP_0}{ d\bbP_1} \right)  d \mathbb{P}_0
= \int \log \left(\frac{ f_0}{f_1} \right) f_0 d \mu .  $$

Finding the density ratio function can be formulated as an optimization problem using the KL divergence. Specifically, we wish to minimize 
 \begin{align}
     Q (f_0) = \int  \log \left( f_0/ f_1 \right) f_{0}  d \mu  \label{qq} 
    \end{align} 
    with respect to  $f_0$ satisfying some moment conditions of the variables that are available from the whole sample.      Let $\bb(\bx)$ be the given balancing score obtained from assumption (\ref{basis}). Because $\bx_i$ are observed throughout the sample, we can estimate  $\mathbb{E}\{\bb(\bX)\}$ from the sample. 
    Given ${\mathbb{E}}\{\bb(\bX)\}$, to utilize the auxiliary information, we have 
\begin{equation}
    p\int  \bb(\bx) f_1 ( \bx) d \mu + (1-p) \int \bb(\bx) f_0 ( \bx) d \mu = \mathbb{E}\{\bb(\bX)\},
    \label{const1}
\end{equation}
where $p = \bbP (\delta = 1)$.   Thus, we wish to find the minimizer of $Q(f_0)$ in (\ref{qq}) subject to the constraint in (\ref{const1}). The following lemma gives the form of the solution to the optimization problem. 


\begin{lemma}
 The minimizer of $Q(f_0)$ in (\ref{qq}) subject to the constraint in (\ref{const1}) is  
\begin{equation} 
{f}_0^\star( \bx) = f_1( \bx) \times \frac{ \exp \{ \blambda_1\trans \bb(  \bx ) \} } { \mathbb{E}_1 \left[ \exp \{ \blambda_1\trans \bb(  \bx ) \} \right]} , 
\label{solution}
\end{equation} 
where $\blambda_1$ is chosen to satisfy  (\ref{const1}). 
\end{lemma}

The solution (\ref{solution}) is obtained using the  information projection  method in information geometry. 
Roughly speaking, constraint in (\ref{const1}) is a moment condition on $f_0( \bx)$ to satisfy (\ref{bdr}). Thus, we can use Lagrange multiplier method to obtain the solution as an exponential tilting form. 
 Note that (\ref{solution}) is equivalent to assuming a model for the density ratio function. That is, our model is 
\begin{equation}
\log \{ r^\star( \bx; \blambda) \} = \lambda_0 +\blambda_1\trans \bb(  \bx ),
\label{loglinear} 
\end{equation}
where $\lambda_0$ is the normalizing constant satisfying $$\int r^\star(\bx; \blambda) f_1( \bx) d \mu = 1.$$ By (\ref{1}), 
 the log-linear density ratio model (\ref{loglinear}) is equivalent to the logistic regression model for the sample selection probability 
\begin{align}
    \mathbb{P} ( \delta = 1 \mid \bx) = \frac{ 1}{ 1+ (1/p-1) \cdot \exp \left\{  \lambda_0 +\blambda_1\trans \bb(  \bx )\right\} }. 
    \label{ps_model}
\end{align} 
Thus, our framework provides another justification of the logistic regression model for the sample selection  probability that  has been considered in \cite{folsom1991exponential}, \cite{kott2006using}, and \cite{tan2020regularized}, among others. \textcolor{black}{But, the parameter estimation discussed in \S 3.3  is different from the existing methods.  Also, the covariates in the parametric model in (\ref{ps_model}) is derived from the outcome model assumption in (\ref{basis}). Nonetheless, including other covariates in the  calibration weighting  can make the resulting PS estimator  robust  against model misspecification (\citealp{han2013estimation,chan2014oracle,chen2017multiply, yang2020doubly}).
}

Using (\ref{solution}), we can express (\ref{const1}) as 
\begin{align}
&p \int     \bb (  \bx) \left[ 1+ \frac{N_0}{N_1}  \cdot  \exp \left\{ \lambda_0 + \blambda_1\trans \bb(  \bx )  \right\} \right] f_1 ( \bx ) d \mu = \mathbb{E}\{ \bb ( \bX)\}  ,\label{const2}
\end{align}
where   $\lambda_0$ satisfies 
\begin{equation} 
 \int \exp \{ \lambda_0 + \blambda_1\trans \bb( \bx ) \}  f_1 ( \bx ) d \mu =1. 
 \label{normal}
 \end{equation} 
The function $ r^{\star}( \bx) \equiv \exp ( \lambda_0 + \blambda_1\trans \bb(  \bx ) ) $ can be viewed  as the projection of the true density ratio function onto the space satisfying  (\ref{const1}).

\subsection{Model parameter estimation} 

 We now wish to estimate the parameters in the log-linear density ratio model in (\ref{loglinear}) and estimate $\btheta$. Note that the parameter is defined through  $\mathbb{E} \{ \bU ( \btheta; \bX, Y) \}=\bzero$, which can be expressed as 
\begin{equation}
p \int  {\bU} ( \btheta; \bx, y)   \left[ 1+ \frac{N_0}{N_1}  \cdot  \exp \left\{ \lambda_0 + \blambda_1\trans \bb(  \bx )  \right\} \right] f_1 ( \bx, y ) d \mu  = 0  \label{const3}
\end{equation}
as an integral equation for $\btheta$, where $\lambda_0$ and $\blambda_1$ satisfy the constraints (\ref{const2}) and (\ref{normal}).

 To estimate the parameters, we use 
 $${\wh{\mathbb{P}}}_1 (x, y) = \frac{1}{  {N_1}  }  \sum_{i=1}^N  {\delta_i } \mathbb{I} \{ (\bx,y) = {(\bx_i, y_i)} \}$$
    to 
    find the minimizer of $D(  {\wh{\mathbb{P}}_1} \parallel   {\mathbb{P}_{1}}  )$ among $\mathbb{P}_{1}$  satisfying the integral constraints  (\ref{const2}), (\ref{normal}), and (\ref{const3}). 
    The problem can be formulated as an optimization problem using 
  the empirical likelihood (EL) method of \cite{qin1994}. In the EL method, a fully nonparametric density $p_i$ can be assumed to replace $f_1( \bx, y)$ in the integral equations. 
  That is, we maximize   $$ \ell( \bp ) = \sum_{i=1}^N  \delta_i \log (p_i) $$
        subject to 
        \begin{equation}
        \sum_{i=1}^N \delta_i p_i=1,
        \label{con1}
        \end{equation} 
      \begin{align}
\frac{N_1}{N} 
\sum_{i=1}^N \delta_i   \bb( \bx_i) \left\{ 1+ \frac{N_0}{N_1}  \cdot  \exp \{ {\lambda}_0 + {\blambda}_1\trans \bb( \bx_i) \}  \right\} p_i  =\frac{1}{N} \sum_{i=1}^N \bb(\bx_i)  , \label{con_3}
\end{align}
 \begin{equation}
\sum_{i=1}^N \delta_i {\bU} ( \btheta; \bx_i, y_i) \left\{ 1+ \frac{N_0}{N_1}  \cdot  \exp \{ {\lambda}_0 + {\blambda}_1\trans \bb( \bx_i) \}  \right\} {p}_i  = \bzero, 
\label{con_4} 
\end{equation}
and 
 \begin{equation} 
  \sum_{i=1}^N \delta_i   \exp \{ \lambda_0 + \blambda_1\trans \bb( \bx_i ) \} p_i=1. 
    \label{con4}
    \end{equation}

The solution to the above optimization problem lead to $\wh{p}_i=\wh{p}_i (  \blambda, \btheta)$  and we obtain the profile empirical likelihood  
\begin{equation} 
 \ell_e ( \btheta, \blambda) = \sum_{i=1}^N \delta_i \log \{ \wh{p}_i ( \blambda, \btheta) \} .
\label{pel}
\end{equation} 
The maximizer of $\ell_e( \btheta, \blambda)$ can be used as  the final estimator of $\btheta$. However, such a joint optimization of $\btheta$ and $\blambda$ is computationally expensive. 

To introduce an alternative computation, 
note that the form of the smoothed propensity weights 
 $$ \wh{\omega}_i^{\star} = \left\{ 1+ \frac{N_0}{N_1} \cdot \exp \left( \wh{\lambda}_0 + \wh{\blambda}_1\trans \bb(\bx_i)\right) \right\} \wh{p}_i $$
 depends on $\btheta$ only through $\wh{p}_i$.  Thus, we can first exclude (\ref{con_4}) in the EL optimization to remove the dependency on $\btheta$ in $\wh{\omega}_i^{\star}$ and then obtain $\wh{\btheta}$ by solving 
$\sum_{i=1}^N \delta_i \wh{\omega}_i^{\star} \bU( \btheta; \bx_i, y_i) = \bzero $ 
        for $\btheta$. This two-step procedure will greatly simplify the computation.

Note that maximizing $ \ell( \bp ) = \sum_{i=1}^N \delta_i \log (p_i) $ subject to (\ref{con1}) only gives $\wh{p}_i=1/N_1$. 
Imposing  (\ref{con_3}) and (\ref{con4}) to the unconstrained optimization does not play any role on the final solution because $\lambda_0$ and $\blambda_1$ are unknown.  Thus, the solution to the  optimization problem remains the same: 
    $$ \wh{p}_i=1/N_1 $$
    with  $\wh{\lambda}_0$ and $\wh{\blambda}_1$ satisfying 
    \begin{align}
        &\sum_{i=1}^N \delta_i
        \left\{ 1+ (N_0/N_1)   \cdot  \exp \{ \wh{\lambda}_0 + \wh{\blambda}_1\trans \bb(  \bx_i ) \}  \right\}  \left[  1, \bb (\bx_i) \right] = \sum_{i=1}^N \left[ 1, \bb (\bx_i) \right] \label{calibration}, 
    \end{align}
    which is a calibration equation for $[ 1, \bb(\bx)]$.

    Once the model parameters in (\ref{loglinear}) are estimated, we can  compute the smoothed propensity score  estimating function 
\begin{equation} 
 \wh{\bU}_{\rm SPS} (\btheta) = \frac{1}{N} \sum_{i =1}^N  \delta_i \wh{\omega}_i^{\star}  \bU( \btheta; \bx_i, y_i) , 
 \label{25}
 \end{equation} 
 where 
 \begin{equation} 
 \wh{\omega}_i^{\star}= 1+ (N_0/N_1)  \cdot  \exp \{ \wh{\lambda}_0 + \wh{\blambda}_1\trans \bb(  \bx_i ) \}  
 \label{swgt}
 \end{equation} 
 and $\wh{\lambda}_0$ and $\wh{\blambda}_1$ are computed from (\ref{calibration}). 
 The final estimator of $\btheta$ is obtained by solving $\wh{\bU}_{\rm SPS} (\btheta)=\mathbf{0}.$

\textcolor{black}{
The proposed propensity score estimating function in (\ref{25}) satisfies the self-efficiency. To see this, without loss of generality, assume that $\btheta$ is a scalar. Note that, for a fixed $\theta$,  we can write 
 \begin{equation} 
  \wh{U}_{\rm SPS} (\theta) 
  = \frac{1}{N} \sum_{i =1}^N\delta_i \wh{\omega}_i^{\star} U_i  =  \frac{1}{N} \sum_{i =1}^N  \left[ {m}_i(\bbeta)  + 
 \delta_i \wh{\omega}_i^{\star} \left\{  U_i - {m}_i( \bbeta) \right\}  \right] 
 \label{26}
 \end{equation} 
 where $U_i = U ( \theta; \bx_i, y_i)$ and 
 $ {m}_i( \bbeta)= \beta_0 + \sum_{j=1}^L \beta_j b_j( \bx_i) 
 $
 for any $\beta_0, \beta_1, \cdots, \beta_L$.  
 Now, since $\wh{\omega}_i^{\star}= 1+ (N_0/N_1)  \cdot  \exp \{ \wh{\lambda}_0 + \wh{\blambda}_1\trans \bb(  \bx_i ) \}  $, the smoothed PS estimator in (\ref{26}) is algebraically equivalent to  
 \begin{align} 
 &\wh{U}_{\rm SPS} (\theta) 
  =  \frac{1}{N} \sum_{i =1}^N  \left\{  \delta_i U_i + (1- \delta_i) m_i({\bbeta})   \right\} 
 + \frac{1}{N}  \sum_{i=1}^N \delta_i  \exp \{ \wh{\lambda}_0 + \wh{\blambda}_1\trans \bb(  \bx_i ) \}  \left\{ U_i -  m_i({\bbeta} )   \right\}. \label{27}
 \end{align} 
 The second term in the right hand side of (\ref{27}) is zero at $\bbeta = \hat{\bbeta}$ where $\hat{\bbeta}$ satisfies 
$$ \sum_{i=1}^N \delta_i  \exp \{ \wh{\lambda}_0 + \wh{\blambda}_1\trans \bb(  \bx_i ) \}  \left\{ U_i -  m_i(\hat{\bbeta} )   \right\} = 0 . 
 $$
Thus, we have 
\begin{equation}
\frac{1}{N} \sum_{i =1}^N\delta_i \wh{\omega}_i^{\star} U_i  = \frac{1}{N} \sum_{i =1}^N  \left\{  \delta_i U_i + (1- \delta_i) {m}_i(\hat{\bbeta})   \right\} .
    \label{28} 
\end{equation}
Therefore, the proposed estimator satisfies  the self-efficiency condition in (\ref{self}). } 
 

Equation (\ref{28}) means that 
 the final inverse propensity  weights $\wh{\omega}_i^{\star}$ do not directly use the regression  model for prediction, but it  implements regression prediction indirectly. Thus, the smoothed PS estimator incorporates the outcome regression model through calibration equation in (\ref{calibration})  and achieves the self-efficiency. 
 On the other hand, model calibration of \cite{wu2001} uses the outcome model directly in the calibration equation.  Similar ideas have been considered in the context of nonparametric calibration estimation in survey sampling. 
  For examples, \cite{montanari2005}  used a single-layer Neural Network model and  \cite{breidt2005model} used penalized Spline  model for nonparametric calibration estimation. 
 
\textcolor{black}{
\begin{remark}
Assumption (\ref{basis}) states that the basis functions do not depend on $\btheta$, which does not always hold. 
If the basis functions in (\ref{basis}) depends on $\btheta$, that is, 
     \begin{equation}
        \mathbb{E}\{ U( \btheta; \bx, Y) \mid \bx \} \in \mbox{span} \{1,  b_1( \bx; \btheta), \ldots, b_L (\bx; \btheta) \} , 
        \label{basis2} 
    \end{equation}
    we can use the following iterative procedure. 
    \begin{enumerate}
    \item Given the current parameter estimate $\hat{\btheta}^{(t)}$, find $\hat{b}_{ki}^{(t)} = b_k ( \bx_{i}; \hat{\btheta}^{(t)} )$, for $k=1, \cdots, L$, such that (\ref{basis2}) holds for $\btheta = \hat{\btheta}^{(t)}$. 
    \item Estimate the parameter in 
     \begin{equation} 
     \omega_i^{(t)} = 1+ (N_0/N_1) \exp \left\{ \phi_0 + \phi_1 \hat{b}_{1i}^{(t)} + \cdots + \phi_L \hat{b}_{Li}^{(t)}   \right\}
     \label{model3}
     \end{equation}
     by apply the calibration equation on $[ 1, b_{1i}^{(t)}, \cdots, b_{Li}^{(t)} ]$. 
     \item Use 
     $$ \hat{\omega}_i^{(t)} = 1+ (N_0/N_1) \exp \left\{ \hat{\phi}_0 + \hat{\phi}_1 \hat{b}_{1i}^{(t)} + \cdots + \hat{\phi}_L \hat{b}_{Li}^{(t)}   \right\}
     $$ 
     to find the solution $\hat{\btheta}^{(t+1)}$ to 
     $$ \sum_{i=1}^N \delta_i  \hat{\omega}_i^{(t)} U( \btheta; \bx_i, y_i) = 0 . $$
     \item Set $t=t+1$ and goto Step 1 until convergence. 
    \end{enumerate} 
    \item This is essentially an extension of the EM algorithm \citep{dempster77} applied to regression model. Step (a)-(b) corresponds to the E-step and Step (c) corresponds to  the M-step. In the E-step, instead of using a parametric model, we use the information projection to compute the conditional expectation  under assumption (\ref{basis}). 
\end{remark}
}

\section{Statistical properties}

We now discuss the asymptotic properties of the smoothed PS estimator $\hat{\btheta}_{\rm SPS}$ which is the solution to $\wh{\bU}_{\rm SPS} (\btheta)=\mathbf{0}$ using $\wh{\bU}_{\rm SPS} (\btheta)$ in (\ref{25}). 
To formally discuss the asymptotic properties of $\hat{\btheta}_{\rm SPS}$,        
 let $\wh{\bU}_2(\blambda)$ be the estimating function for ${\blambda}=({\lambda}_0, {\blambda}_1)$.  By  (\ref{calibration}), we can express 
\begin{equation} 
 \wh{\bU}_2(\blambda) = \frac{1}{N} \sum_{i=1}^N \delta_i  \omega^{\star} (\bx_i; \blambda) [1, {\bb}(\bx_i)]  - \frac{1}{N}  
\sum_{i=1}^N  [1, {\bb}(\bx_i) ], 
\label{est3}
\end{equation} 
where $\omega^{\star} (\bx; \blambda) = 1+ (N_0/N_1) \exp \{ \lambda_0 + \blambda_1\trans \bb(\bx) \}$. 
Let $  \blambda^{\star}$ be the true parameter value in the log-linear density ratio model in (\ref{loglinear}) so that 
$ \exp\{ \lambda_0^{\star}+  \bb\trans(\bx) \blambda_1^{\star}\}  = {r}^{\star}( \bx) $. By the constraint in  (\ref{const1}), we have $\mathbb{E} \{\wh{\bU}_2(\blambda^{\star}) \} =\bzero$, where the reference distribution is the conditional distribution of $\bx$ given $\delta$.   The unbiasedness of $\wh{\bU}_2 ( \blambda^{\star})$ can also be derived under the selection  probability model associated with (\ref{loglinear}).  That is, under the selection probability  model 
\begin{align}
    \mathbb{P} ( \delta = 1 \mid \bx) = \frac{ \exp \{ \lambda_0^{\star} + {\bb}\trans (\bx) \blambda_1^{\star} \}}{ 1+ \exp \{ \lambda_0^{\star} + {\bb}\trans (\bx) \blambda_1^{\star}  \} }:= {\pi}^{\star}(\bx), \label{RP}
\end{align} 
 we can also obtain $\mathbb{E} \{\wh{\bU}_2(\blambda^{\star}) \} =\bzero$. 

Thus, as long as the sufficient conditions for $\mathbb{E} \{\wh{\bU}_2(\blambda^{\star}) \} = \bzero$ are satisfied, 
we can establish the weak consistency of $\wh{\blambda}$ and apply the standard Taylor linearization to obtain the following theorem, where $\blambda^{\star}$ is 
the true value of parameter in model (\ref{loglinear}).  

The regularity conditions and the proof are presented in the Supplementary Material. 
\begin{theorem}\label{T1}
Assume that density ratio model (\ref{loglinear}) holds. Under the regularity conditions described in the Supplementary Material, we have
\begin{align}
 \wh{\bU}_{\rm SPS} (\btheta) = \frac{1}{N} \sum_{i=1}^N \left[  \bz_i\trans \bbeta^{\star} + \delta_i \omega^{\star} ( \bx_i; \blambda^{\star})  \{ \bU ( \btheta; \bx_i, y_i)-\bz_i\trans \bbeta^{\star} \}  \right] + o_p(N^{-1/2}), \label{33} 
\end{align}
where $\bz_i\trans=(1, \bb\trans(\bx_i)) $, $\omega_i^{\star} (\bx; \blambda) = 1+ (N_0/N_1) \exp \{ \lambda_0 + \blambda_1\trans \bb(\bx) \}$,  
and $\bbeta^{\star} = \bbeta^{\star}(\btheta)$ is the probability limit of the solution to 
\begin{equation} 
 \sum_{i=1}^{N} \delta_{i}  \exp( \bz_i\trans \blambda^{\star} )\left\{ \bU ( \btheta; \bx_i, y_i)  - \bz_i'  \bbeta\right\} \bz_i=\bzero. 
 \label{T12b}
 \end{equation} 
 
 \end{theorem}

By Theorem~\ref{T1}, we can obtain, ignoring the smaller order terms, 
\begin{align} 
\wh{\bU}_{\rm SPS} ( \btheta) - \wh{\bU}_{\rm N} (\btheta)  
= \frac{1}{N} \sum_{i=1}^N   \{  \delta_i  \omega^{\star} ( \bx_i; \blambda^{\star}) -1 \}  \{ \bU ( \btheta; \bx_i, y_i) -\bz_i\trans \bbeta^{\star} \}\label{32}
\end{align} 
where $\wh{\bU}_{\rm N} ( \btheta) =N^{-1} \sum_{i=1}^N \bU ( \btheta; \bx_i, y_i)$. Under the selection probability model (\ref{RP}), the first term of (\ref{32}) has zero expectation. Also, 
 if the regression  outcome model satisfies 
\begin{equation} 
 \mathbb{E} \{ \bU (\btheta; \bx, Y) \mid \bx \}  \in \mathcal{H} \equiv \mbox{span} \{ 1, b_1( \bx), \cdots, b_L(\bx) \} , 
\label{assume1} 
\end{equation}
then we obtain $\bz_i\trans \bbeta^{\star} =  \mathbb{E} \{ \bU (\btheta; \bx_i, Y) \mid \bx_i \}$ and the second term of (\ref{32}) has zero expectation. Thus, expression in (\ref{32}) gives the doubly robust property  \citep{bang2005doubly,tsiatis2007semiparametric,cao2009improving, kimhaziza12}  in that the resulting estimator is justified either the outcome regression model or the selection probability model is correctly specified.

\begin{corollary}
Suppose that the  assumptions for Theorem \ref{T1} hold.
 If $\bar{\bU} (\btheta; \bx) = \mathbb{E} \{ \bU (\btheta; \bx, Y) \mid \bx \}$ satisfies (\ref{assume1}), 
 we obtain $\bar{\bU} (\btheta; \bx) = \bz_i\trans \bbeta^\star$ and 
 \begin{equation}
     \sqrt{N}\left( \wh{\btheta}_{\rm SPS} - \btheta_0  \right) \stackrel{\mathcal{L}}{\longrightarrow} N(0,\btau^{-1}  V_1 (\btau^{-1 })\trans ),
     \label{result1}
 \end{equation}
 where $\btau = \mathbb{E}\{ \partial \bU( \btheta_0; \bX, Y)/ \partial \btheta\trans\}$, 
 \begin{equation} 
 V_1 = \mathbb{V} \left\{ \bar{\bU} (\btheta_0; X) \right\} + \mathbb{E} \left[  \delta \{ \omega^{\star}(\bX; \blambda^{\star} )\}^2  \mathbb{V} \{ U( \btheta_0; \bX, Y)  \mid \bX \} \right] ,
\label{var1}
\end{equation} 
and 
$ \omega^{\star}(\bX; \blambda ) = 1+ (N_0/N_1) \exp \{ \lambda_0 + \blambda\trans \bb( \bx) \} . $ 
\label{C2} 
\end{corollary}
 

 \begin{remark}[Remark 2]
 The variance term in (\ref{var1}) deserves a further discussion. 
Let $\pi( \bx)$ be the true selection probability.  We may use the true  selection probability to construct a doubly robust estimator of the form  
 \begin{align}
\wh{\bU}_{\rm DR}(\btheta) 
=  \frac{1}{N} \sum_{i=1}^N \left[ \bar{\bU} (\btheta; \bx_i) + \frac{ \delta_i}{ \pi (\bx_i) }  \{ \bU( \btheta; \bx_i, y_i) -\bar{\bU} (\btheta; \bx_i) \} \right]. \label{double2}
\end{align} 
 Now, by (\ref{1}) and (\ref{bdr}), we have 
 $$ \omega^{\star}( \bx) = \mathbb{E}\left\{ \frac{1}{ \pi( \bx)} \mid \bb( \bx) , \delta =1  \right\} .
 $$
 By Jensen's inequality, we can obtain 
\begin{align} 
 &V \{ \wh{\bU}_{\rm DR}(\btheta)  \}  \label{var}\\
 =& \mathbb{V} \left\{ \bar{\bU} (\btheta; \bX) \right\} + \mathbb{E} \left[  \delta \{ \pi(\bX )\}^{-2}  \mathbb{V} \{ \bU( \btheta; \bX, Y)  \mid \bX \} \right] \notag \\
 \ge& \mathbb{V} \left\{ \bar{\bU} (\btheta; \bX) \right\} + \mathbb{E} \left[  \delta \{  \omega^{\star}( \bx)  \}^{2}  \mathbb{V} \{ \bU( \btheta; \bX, Y)  \mid \bX \} \right] \notag = V \{ \wh{\bU}_{\rm SPS}(\btheta)  \}.\notag
 \end{align} 
 Thus, the smoothed PS estimator is more efficient than the doubly robust estimator in (\ref{double2}). 
  Furthermore, if the selection  probability model satisfies  (\ref{RP}),  we have equality in (\ref{var}) and the smoothed  PS estimator is locally optimal in the sense that it achieves the variance lower bound 
  of \cite{robins1994estimation}. 
 \end{remark}

 \begin{remark}[Remark 3]
 
 Condition (\ref{assume1}) is a critical condition for the validity of the proposed estimator.  If the space $\mathcal{H} = \mbox{span} \{ \bb(\bx) \}$ is large enough, then (\ref{assume1}) is likely to be satisfied and result (\ref{result1}) will hold. However, if $\mathcal{H}$ is too large, then we can find $\mathcal{H}_0 \subset \mathcal{H}$ such that $\mathbb{E}( Y \mid \bx)   \in \mathcal{H}_0$. In this case, we can construct a smoothed density ratio function using the basis functions in $\mathcal{H}_0$ only. By Remark 2, it is more efficient than the smoothed PS estimator using the basis function in $\mathcal{H}$. Therefore, including unnecessary calibration variables in the calibration equation will increase the variance. This is consistent with the empirical findings of 
 \cite{brookhart2006variable} and  \cite{shortreed2017}. 
 We will discuss this result further in Section \ref{s4}. See also the second simulation study in Section 7.2. 
 \end{remark}

We now discuss variance estimation.  Using (\ref{33}), we obtain 
\begin{equation}
\wh{\btheta} - \btheta_0 
= - \btau^{-1} \frac{1}{N} \sum_{i=1}^N \bd( \bx_i, y_i, \delta_i; \btheta_0, \blambda^{\star}) + o_p( N^{-1/2}), 
\label{T12} 
\end{equation} 
where $\btau = \mathbb{E}\{ \partial \bU( \btheta_0; \bX, Y)/ \partial \btheta\trans\}$, 
\begin{align*} 
\bd( \bx_i, y_i, \delta_i; \btheta_0, \blambda^{\star}) 
=  \sum_{k=0}^L  b_k( \bx_i) \bbeta_k^{\star} +  \delta_i{\omega}^{\star} (\bx_i;  {\blambda}^{\star})  \{ \bU( \btheta_0; \bx_i, y_i) -  \sum_{k=0}^L  b_k( \bx_i) \bbeta_k^{\star}  \} \notag, 
\end{align*}
and  $\bbeta_k^{\star}$ are the solution to 
\begin{align*}  &\sum_{i=1}^N \delta_i \exp \{ {\lambda}_0^\star+ \blambda_1^{\star \trans} \bb(\bx_i) \}  
\times\left\{ \bU( \btheta_0; \bx_i, y_i)  -  \sum_{j=0}^L  b_j( \bx_i) \bbeta_j^{\star} \right\} b_k ( \bx_i)= 0 
\end{align*}
for all $k=0,1, \cdots, L$.

Now, the  variance estimation of $\wh{\btheta}_{\rm SPS}$ can be constructed from (\ref{T12}). Specifically, let 
\begin{align*}
  \wh{\bd}_{i}  & = \sum_{k=0}^L  b_k( \bx_i)  \wh{\bbeta}_k 
  + 
    \delta_{i} \left\{ 1+ 
    \frac{N_0}{N_1} \cdot r^{\star}( \bx_i; \wh{\blambda})  \right\} \left\{\bU( \wh{\btheta}; \bx_i, y_i) 
    - \sum_{k=0}^L  b_k( \bx_i) \hat{\bbeta}_k \right\},
\end{align*}
where  $\wh{{\bbeta}}_k$ are the solution to 
\begin{align*} 
\sum_{i=1}^N \delta_i \exp \{ \hat{\lambda}_0+ \wh{\blambda}_1^{ \trans} \bb(\bx_i) \} \left\{ \bU( \wh{\btheta}; \bx_i, y_i)  -  \sum_{j=0}^L  b_j( \bx_i) \wh{\bbeta}_j \right\} b_k ( \bx_i)  
= 0 
\end{align*}
for all $k=0,1, \cdots, L$.  
The linearization variance estimator is then  written as
\begin{align}
    \wh{\mbox{V}}(\wh{\btheta}_{SPS})  =\frac{1}{N}  \hat{\btau}^{-1} \wh{\Sigma}_{dd} \hat{\btau}^{-1 '}  
    ,\notag
\end{align}
where $\hat{\btau}= N^{-1} \sum_{i =1}^N \delta_i \wh{\omega}_i^{\star} \dot{\bU}( \hat{\btheta}; \bx_i, y_i)$, $\dot{U}( \btheta; \bx, y)= \partial \bU( \btheta; \bx, y)/ \partial \btheta\trans$,  
$\wh{\Sigma}_{dd}= (N-1)^{-1} \sum_{i=1}^{N}(\wh{\bd}_{i} - \wb{\wh{\bd}}_{N})^{\otimes 2} $ and 
 $\wb{\wh{\bd}}_{N} = N^{-1} \sum_{i=1}^{N}\wh{\bd}_{i}$.

The empirical likelihood approach also provides a way to develop a likelihood ratio  test   for parameters in a completely analogous way to that for parametric likelihoods. 
Let  $\ell_e(\btheta, \blambda)$ be the profile empirical log-likelihood in (\ref{pel}) under constraints \eqref{con1}, \eqref{con_3} and \eqref{con4}. Further, define the profile empirical  likelihood 
\begin{align}
    \ell_{p}(\btheta) = \max_{ \blambda } \ell_{e}(\btheta, \blambda). 
\end{align}
We have the following result. 
\begin{theorem}
Under regularity conditions in the Supplementary Material, 
\begin{align}
    2\{\ell_{p}(\wh{\btheta}) - \ell_{p}(\btheta_{0}) \}  \stackrel{\mathcal{L}}{\longrightarrow} \chi_{p}^{2},
\end{align}
as $N \rightarrow \infty$, where $\wh{\btheta}$ is obtained from the two-step procedure and $\btheta_{0}$ is the true parameter. 
\end{theorem}
For $\btheta=(\btheta_1, \btheta_2),$ the empirical likelihood ratio test for $H_0: \btheta_1= \btheta_1^0$ can also be developed similarly to \cite{qin1994}. 

\section{Dimension Reduction} \label{s4}

We now consider the case of $\theta=\mathbb{E}(Y)$. 
In Section 4, {we have seen that   either $\mathbb{E}(Y \mid \bx)$ lies in the linear space $\mathcal{H}=\mbox{span} \{ \bb(\bx) \}$ or  $\bb(\bx)$ is the basis for logistic regression \eqref{RP} can give the root-$n$-consistency of the proposed estimator.}  However, as pointed out in Remark 3, including other $\bx$-variables outside the outcome model into 
$\mathcal{H}$ may lead to efficiency loss.

To explain the idea further, we assume that {$\bb(\bx) = \bx_{\mathcal{M}}$, where $\mathcal{M}$ is an index set for a subset of $\bx$. }
The following lemma presents an interesting result.
\begin{lemma}
 If MAR condition in (\ref{mar}) holds  and the reduced model for $y$ holds such that 
    \begin{equation} 
    f( y \mid \bx) = f( y \mid \bx_{\mathcal{M}})
    \label{reduced}
    \end{equation} 
    for $\bx= (\bx_{\mathcal{M}}, \bx_{\mathcal{M}^c})$, then we can obtain the reduced MAR  given $\bx_{\mathcal{M}}$. That is,
    \begin{equation} 
     Y \perp \delta \mid \bX_{\mathcal{M}} . 
     \label{rmar}
     \end{equation} 
    \label{lemma4.1} 
\end{lemma}
Note that (\ref{rmar}) is a special case of the reduced MAR in (\ref{mar2}) using $\bb(\bx) = \bx_{\mathcal{M}}$ as the balancing score function. In the spirit of Remark 3, we can see that the smoothed PS estimator using $\bb(\bx)=\bx_{\mathcal{M}}$ is more efficient than the PS estimator using $\bb(\bx) = \bx$. 
Therefore, it is better to apply a model selection procedure to select the important variables
which satisfies \eqref{reduced}.

To find $\bx_{\mathcal{M}}$ satisfying \eqref{reduced}, we utilize two-stage estimation strategy to complete the smoothed  propensity score function:
\begin{enumerate}
  \item Step 1: Use a penalized regression method  to select the basis function for the regression of $y$ on $\bx$. 
  \item Step 2: Use the basis function in Step 1 to 
  to obtain the calibration equation in \eqref{calibration} and construct $\wh{r}( \bx) = r^{\star}( \bx; \wh{\blambda})= \exp \{ \hat{\lambda}_0+  {\bb}\trans(\bx_i) \wh{\blambda}_1 \}$. 
\end{enumerate}

In the first stage, we adapt the penalized estimating equations \citep{johnson2008penalized} to select important covariates in the outcome model.  Generally speaking, we utilize an $Z$-estimator in the  outcome model and we denote the corresponding score function as $\bU(\balpha)$. For example, $\bU(\balpha)$ can be written as
\begin{align}
    \bU(\balpha) = \frac{2}{N_{1}}\sum_{i=1}^{N}\delta_{i}\bx_{i}( \bx_{i}\trans\balpha - y_{i}),
\end{align}
where $\balpha \in \mathbb{R}^{d+1}$. The penalized estimating equations can be written as
\begin{align}\label{U P}
    \bU^{P}(\balpha) = \bU(\balpha) - q_{\lambda}(|\balpha|)\mbox{sgn}(\balpha),
\end{align}
where $q_{\lambda}(|\balpha|) = (q_{\lambda}(|\alpha_0|), \ldots,  (q_{\lambda}(|\alpha_{d}|) )\trans$, $q_{\lambda}(\cdot)$ is a continuous function and $q_{\lambda}(|\balpha|)\mbox{sgn}(\balpha)$ is an elementwise product between two vectors. Further, let $p_{\lambda}(x) = \int q_{\lambda}(x)dx$. In M-estimation framework, $p_{\lambda}(x)$ usually serves as a penalization function. Various penalization functions can be used but we only consider  the smoothly clipped absolute deviation function (SCAD) \citep{fan2001variable}. In particular, 
\begin{align}
    q_{\lambda}(\alpha) = \lambda\left\{ \mathbb{I}(|\alpha|<\lambda) +
\frac{ (a\lambda - |\alpha|)_{+}  }{(a-1)\lambda} \mathbb{I}(|\alpha| \geq \lambda)     \right\},
\end{align}
where $(x)_{+} = \max\{x, 0\}$, $\mathbb{I}$ is the indicator function and $a$ is constant specified as $3.7$ in \cite{fan2001variable}. Further, we let $\wh{\mathcal{M}}$ to denote the variable index set selected after SCAD procedure.  Here we use a working model to select the variables. 
According to   \cite{johnson2008penalized}, under some regularity conditions, we have
\begin{align}
 \mathbb{P}(\wh{\alpha}_{j} \neq 0) \rightarrow 1, \ \mbox{for}\ j \in \mathcal{M};\notag\\
    \mathbb{P}(\wh{\alpha}_{j} = 0) \rightarrow 1, \ \mbox{for}\ j \in \mathcal{M}^{c},\label{model selection consistency}
\end{align}
which constructs the model selection consistency. After the first stage, we now obtain the important variable set $\wh{\mathcal{M}}$.

In the second stage, we use the selected variables to compute the smoothed density ratio function. That is, use 
\begin{equation*}
\log \{ r^\star( \bx_{\mathcal{M}}) \} = \lambda_0 +\blambda_1\trans \bx_{\mathcal{M} } . 
\end{equation*}
 The resulting PS estimator is then computed by 
 $$\wh{\theta}_{\rm SPS} = N^{-1} \sum_{i =1}^N\delta_i \wh{\omega}^{\star} (\bx_i) y_i $$
 with $\wh{\omega}^{\star}(\bx) = 1+ (N_0/N_1)\exp ( \bx_{\mathcal{M}}\trans \wh{\blambda} )$. 

\begin{corollary}\label{C1}
Suppose that the assumptions for Theorem \ref{T1} hold. Also, the additional assumptions  listed in the Supplementary Material  hold.  
If $\bx_{\mathcal{M}}$ satisfies $\mathbb{E}( Y \mid \bx ) = \mathbb{E}( Y \mid \bx_\mathcal{M})$, then we obtain, {with probability goes to 1,} 
 \begin{equation}
     \sqrt{N}\left( \wh{\theta}_{\rm SPS} - \theta_0  \right) \stackrel{\mathcal{L}}{\longrightarrow} N(\bzero, V_r ),
     \label{result4}
 \end{equation}
 where 
 \begin{align}
V_r =& \mathbb{V}  \left\{\mathbb{E}(Y\mid \bX_{\mathcal{M}})\right\} + \mathbb{E}\left[
\delta \{ 1+ (N_0/N_1)  r^{\star}( \bX_\mathcal{M})  \}^2 \mathbb{V} (Y\mid \bX_{\mathcal{M}})   \right]\notag.\label{var2}
\end{align}
 \end{corollary}

By Corollary \ref{C1}, we can safely ignore the uncertainty due to the model selection in the first step. The asymptotic results in \eqref{result4} is based
 on the model selection consistency in \eqref{model selection consistency}. {That is, if we define $\mathcal{D}_n = \{ \mathcal{M} = \wh{\mathcal{M}} \}$ where $\wh{\mathcal{M}}$ is obtained from the first stage procedure, 
 the linearization in (\ref{result4}) is conditional on $\mathcal{D}_n$.  By (\ref{model selection consistency}), we obtain $P( \mathcal{D}_n ) \rightarrow 1$ and the limiting distribution of $T_n \equiv \sqrt{n}( \wh{\theta}_{\rm SPS} - \theta )$ given  $\wh{\mathcal{M}}$, denoted by $\mathcal{L} (T_n \mid \wh{\mathcal{M}} )$,  is asymptotically equivalent to  $\mathcal{L} (T_n \mid \mathcal{M} )$.  
 See also Theorem 1 of \cite{yang2020doubly} for a similar argument. 
 }

\section{Multivariate missing data}\label{ss62}

We now consider the case of multivariate study variables, denoted by $Y_1, \cdots, Y_p$, and they are subject to missingness. There are $2^p$ possible missing patterns with $p$ study variables. 
 Let $T \le 2^p$ be the realized number of different missing patterns in the sample. Thus, the sample is partitioned into $T$ disjoint subsets with the same missing patterns. The parameter of interest is defined through $\mathbb{E} \{ \bU ( \btheta ; \mathbf{Y} ) \}=0$, where $\bY = (Y_1, \ldots, Y_p)\trans$.  

 Let $S_t$ be the $t$-th subset of the sample from this partition. We assume that $S_1$ consists  of elements with complete response and that $S_1$ is nonempty.
 Without loss of generality, we may define $\delta_{i,t}=1$ if $i \in S_t$ and $\delta_{i, t}=0$ otherwise. 
 We wish to construct an estimating function using all available information: 
 \begin{eqnarray*} 
\bar{\bU} ( \btheta ) 
&=& \sum_{t=1}^T \sum_{i \in S_t}  \mathbb{E}\{ \bU(\btheta; \by_i) \mid \by_{i, obs(t) } \} 
\end{eqnarray*} 
where $\by_{i, obs(t)}$ is the observed part of $\by_i$ for $i \in S_t$.  
 Instead of using a model for each conditional distribution, we can use the density ratio model such that 
\begin{equation} 
N_1^{-1} \sum_{i \in S_1} r_t^{\star}( \by_{i, obs(t)}) \bU( \btheta;  \by_i) = N_t^{-1} \sum_{i \in S_t} \mathbb{E}\{ \bU( \btheta; \by_i) \mid \by_{i, obs(t)} \} 
\label{con5}
\end{equation}
for $t=2,\ldots, T$. \textcolor{black}{Condition (\ref{con5}) is  the self-efficiency condition in (\ref{self}). } 
    
 To construct the density ratio function satisfying (\ref{con5}), we first find $\mathcal{H}_t = \mbox{span}\{ b_1^{(t)}(\by_{obs(t)}),$ $ \cdots , b_{L(t)}^{(t)} (\by_{obs(t)})  \}$ such that $\mathbb{E}\{ \bU(\btheta; \by_i) \mid \by_{i, obs(t) } \}  \in \mathcal{H}_t$.  
Thus, using the I-projection idea in Section 3,  we may assume 
\begin{equation}
\log \{ r_t^{\star} (\by_{obs(t)}; {\bphi^{(t)}} ) \} = { \phi_0^{(t)} }+ \sum_{j=1}^{L(t)} {  \phi_j^{(t)} } b_j^{(t)} (\by_{obs(t)} ) 
\label{model5} 
\end{equation}
as the log-linear model for density ratio function.  
 The model parameters in (\ref{model5}) can be estimated by calibration equation derived from (\ref{con5}) and model assumption (\ref{model5}): 
        \begin{eqnarray*}
   N_1^{-1} \sum_{i \in S_1} r_t^{\star}( \by_{i, obs(t)}; {\bphi^{(t)}})  ( 1, \bb_i^{(t) }  ) &=& N_t^{-1} \sum_{i \in S_t}  ( 1, \bb_i^{(t) }  )  
\end{eqnarray*}
with respect to ${\bphi^{(t)}}$ for $t=2,\cdots, T$, where $\bb_i^{(t)}$ is a vector of $b_j^{(t)} (\by_{i, obs(t)})$ for $j=1, \cdots, L(t)$.

Now, the smoothed PS estimator of $\btheta$  can be obtained by solving 
\begin{equation} 
\wh{\bU}_{\rm SPS} ( \btheta) \equiv N^{-1}  \sum_{i \in S_1} \wh{\omega}_i^{\star} \bU ( \btheta;  \by_i) = \bzero , 
\label{ps3}
\end{equation} 
where  
\begin{equation*}
 \wh{\omega}_i^{\star} = \sum_{t=1}^{T} \frac{N_t}{N_1} r^{\star}( \by_{i, obs(t)}; \wh{\bphi}^{(t)} )
 \end{equation*} 
 is the final weights for PS estimation.

To investigate the asymptotic behavior of the solution to (\ref{ps3}), 
the density ratio for missing pattern $t$ can be simplified as $r_{t}^{\star}(\bz_{i,t}; \blambda_{t})$, for $t = 2, \ldots, T$.
Let the true parameter of interest be $\btheta_{0}$ and the true parameter for density ratio be $\blambda_{0} = (\blambda_{2}\trans, \ldots, \blambda_{T}\trans)\trans$. Then we have the following theorem.

\begin{theorem}\label{TMM}
Under the regularity conditions stated in the Supplemtentary Material, for multivariate missing case, the solution $\wh{\btheta}_{\rm SPS}$ to (\ref{ps3}) satisfies 
\begin{align}
 \sqrt{N} \left(\wh{\btheta}_{\rm SPS} - \btheta_{0}\right) = \frac{1}{\sqrt{N}} \sum_{i=1}^N d( \bx_{i}, \by_{i}, \delta_i ; \blambda_{0}) + o_p (1 ),\notag
 \end{align}
 where 
 \begin{align}
    &d( \bx_{i}, \by_{i}, \delta_i ; \blambda_{0})  \\
    =& - \left[\mathbb{E}\left\{      \frac{\partial}{\partial \btheta }   \bU(\btheta_{0}; \bx, \by)    \right\}\right]^{-1} \times   \left[    \delta_{i,1}\bU_i 
 + \sum_{t=2}^{T}{\delta}_{i,t} \wt{\bbeta}_{t}\bz_{i,t} + {\delta}_{i,1} \sum_{t=2}^{T} \frac{N_{t}}{N_{1}} r_{t}^{\star}(\bz_{i,t}; \blambda_{t}) \left\{ \bU_i- \wt{\bbeta}_{t} \bz_{i,t}   \right\} 
     \right] ,\notag
\end{align} 
$\bU_i = \bU( \btheta; \bx_i,  \by_i)$, $\bz_{i,t}=(1, \bb_i^{(t)})$ and 
$\wt{\bbeta}_{t}$ is the probability limit to the solution of 
\begin{align}
\sum_{i=1}^{N} {\delta}_{i,1}r_{t}^{\star}(\bz_{i,t};\blambda_{t})  \left\{\bU(\btheta_{0}; \bx_{i}, \by_{i}) - \bbeta_{t} \bz_{i,t}\right\} \bz_{i,t}\trans = \bzero.\notag
\end{align}
\end{theorem}


The above theorem depicts the asymptotic behavior of our proposed estimators in multivariate missing case. Its proof is presented in the Supplementary Materials.

\section{Simulation Study}

\subsection{Simulation Study One}\label{sssimulation1}

Two limited simulation studies are  performed to check the performance of the proposed method. 
In the first simulation study, we compare the proposed method with other existing methods under the MAR setup with few covariates. 
 The setup for the first simulation study employed  a $2 \times  2$ factorial structure with two factors. The first factor is the outcome regression (OR) model that generates the sample. The second factor is the response mechanism (RM). We generate $\delta$ and $\bx = (x_{1}, x_{2}, x_{3}, x_{4})\trans$ based on the RM first. We have two different setup for the response mechanism as follows: 
 \begin{itemize} 
\item  RM1 (Logsitic model):\begin{align}
    x_{ik} &\sim N(2,1), \mbox{for\ }k = 1, \ldots, 4,\notag\\
   \delta_{i} &\sim \mbox{Ber}(p_{i}),\notag\\
   \mbox{logit}(p_{i}) &=  1 - x_{i1} + 0.5x_{i2} +  0.5x_{i3} - 0.25x_{i4}.\notag
\end{align}
 \item RM2(Gaussian mixture model):\begin{align}
   \delta_{i} &\sim \mbox{Bern}(0.6)\notag\\
    x_{ik} &\sim N(2,1), \mbox{for\ }k = 1, 2, 3,\notag\\
  x_{i4} &\sim\begin{cases}
  N(3,1), \mbox{if\ }\delta_{i} = 1\notag\\
  N(1,1), \mbox{otherwise}.
  \end{cases}
\end{align}
  \end{itemize}

 Once $x$ and $\delta$ are generated, we  generate $y$ from two different outcome models, OR1 and OR2, respectively. That is, we generate $y$ from 
 \begin{itemize} 
\item  OR1: $y_i = 1 +x_{i1} + x_{i2} + x_{i3} + x_{i4} + e_i. $
\item OR2: $y_i = 1 +0.5 x_{i1}  x_{i2} + 0.5 x_{i3}^{2} x_{i4}^{2} + e_i. $ 
  \end{itemize} 
 Here, $e_i \sim N(0, 1)$.

 From each of the sample, we compare four different PS  estimators of $\theta$:
\begin{enumerate}
\item The proposed information projection (IP) PS estimator using the maximum  entropy method in Section 3 using $\bx_i = (1, x_{i1}, x_{i2}, x_{3i}, x_{i4})\trans$ as the control variable for calibration. Thus, the proposed PS estimator can be written as $\wh{\theta}_{{\rm SPS}}= N^{-1} \sum_{i=1}^N \delta_i \wh{\omega}_i^\star y_i$ where $\wh{\omega}_i^\star$ is of the form in  (\ref{swgt}) with 
        \begin{equation} 
        \sum_{i=1}^N \delta_i \wh{\omega}_i^\star (1, x_{i1}, x_{i2}, x_{3i}, x_{i4}) =\sum_{i=1}^N (1, x_{i1}, x_{i2}, x_{3i}, x_{i4}) .  
        \label{cal5}
        \end{equation} 
    \item The classical PS estimator using 
    maximum likelihood estimation (MLE) of the response probability with Bernoulli distribution with parameter $\mbox{logit}(p_{i}) = \bx_{i}\trans \blambda $.
    \item Covariate balancing propensity score method (CBPS) of  \citet{imai2014covariate}. The CBPS estimator of $\theta$ is obtained by applying empirical likelihood method that maximizes  $ l( \omega) = \sum_{i=1}^n \delta_i   \log (\omega_i) $ 
        subject to (\ref{cal5}). 
    \item Entropy balancing propensity score method (EBPS) of  \citet{hainmueller2012entropy}  using calibration variable $(1, x_{1}, x_{2}, x_{3}, x_{4})\trans$. That is, find the weights $\omega_i$ that maximize 
        $ Q( w) = \sum_{i=1}^n \delta_i \omega_i \log (\omega_i) $
        subject to (\ref{cal5}). 
\end{enumerate}



To check the performance of the four estimators, 
 we use sample size $N=5,000$ with $5,000$ Monte Carlo samples. The results are presented in Figure \ref{fig:MAR}, where IP is our proposed method.

\begin{figure}
         \centering
         \includegraphics[width=1\textwidth]{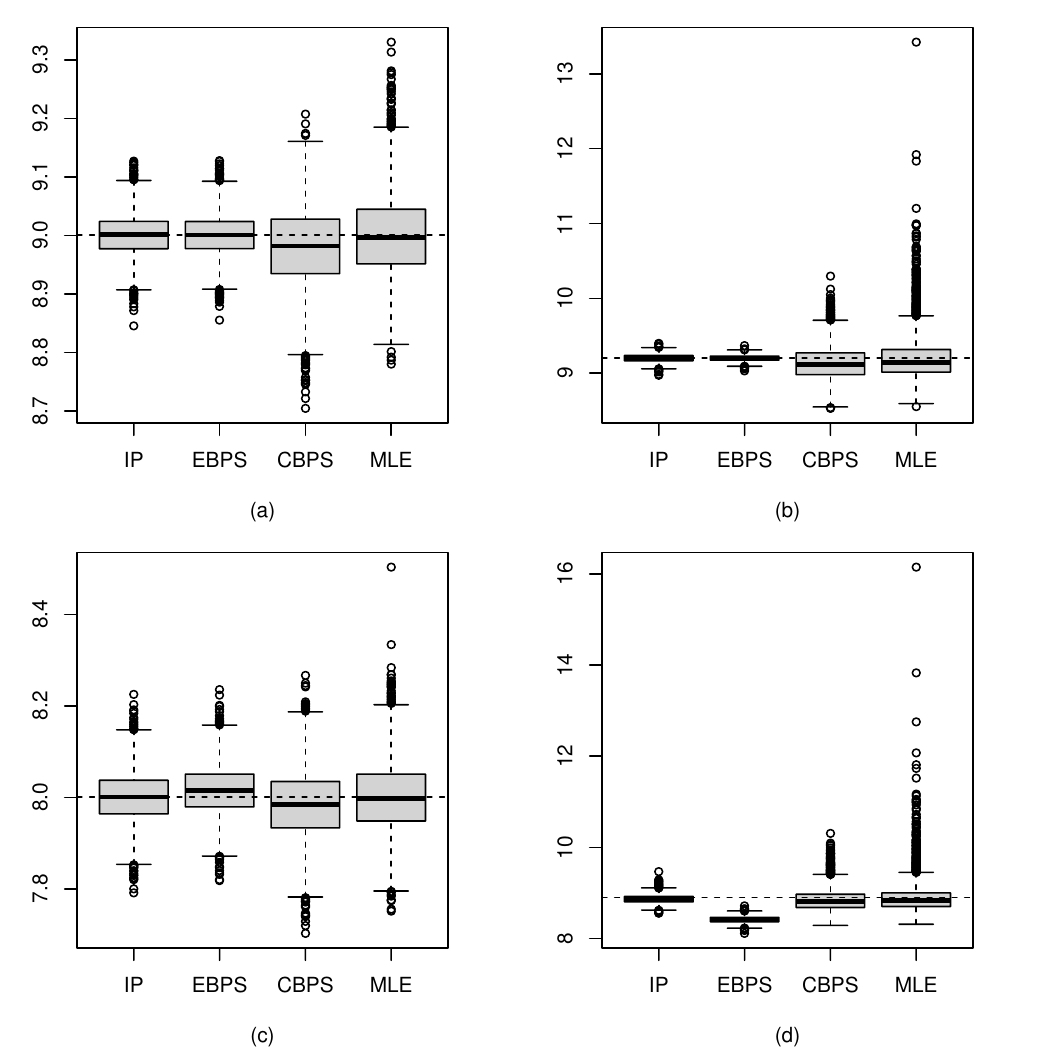}
         \caption{Boxplots with four estimators for four models under simulation study one: (a) for OR1RM1, (b) OR1RM2, (c) for  OR2RM1 and (d) for OR2RM2.}
         \label{fig:MAR}
\end{figure}

 When we use $(1, x_1, x_2, x_3, x_4)$ as the calibration variable, OR1 matches with the working outcome model and RM1 matches with the working response model. Among the four methods considered, our proposed method and EBPS method perform better than the other two methods. 
Entropy balancing propensity score method also shows good performances  when OR1 or RM1 is true as it is doubly robust, but when both models fail (i.e. OR2RM2 setup),  the performance is really poor.  Our proposed method is also doubly robust and it performs reasonably well even when both models fail. Note that our proposed PS weights can be written as 
$$ \hat{\omega}_i^\star=  1+ {N_0} \times  \frac{ \exp ( \bx_i\trans \wh{\blambda})}{ \sum_{k=1}^N \delta_k \exp ( \bx_k\trans \wh{\blambda}) }  . 
$$
Thus, the PS estimator of $\theta=\mathbb{E}(Y)$ can be expressed as 
$$ \wh{\theta}_{\rm SPS} = \frac{1}{N}  \sum_{i=1}^N \left\{ \delta_i y_i + (1- \delta_i) \wh{y}_i  \right\} $$
where 
\begin{equation} 
 \wh{y}_i =  \frac{ \sum_{j=1}^N \delta_j \exp ( \bx_j\trans \wh{\blambda}) y_j }{ \sum_{k=1}^N \delta_k \exp ( \bx_k\trans \wh{\blambda}) }.\notag
\end{equation} 
Thus, the proposed PS estimator can be expressed as a special case of fractional hot deck imputation of  \cite{kimyang14}, which is a robust estimation method.  
Other PS estimators does not allow for this interpretation.

\subsection{Simulation Study Two}

We performed another simulation study to understand the effect of the calibration variables 
for smoothed PS weighting. We first generate  $\bx = (x_1, x_2, x_3)$ from multivariate normal distribution with mean $(1, 1, 1)$ and variance $\Sigma$ where
$$
\Sigma = \begin{pmatrix}
1 & 0 & 0 \\
0 & 1 & 0.5 \\
0 & 0.5 & 1
\end{pmatrix}
$$
and generate $y$ from $y = 1+ 0.5 x_1 - x_2 + e$ where $e \sim N(0, \sigma^2 = 1)$. We also generate $\delta_i \sim \text{Bernoulli}(\pi_i)$ where $\text{logit}(\pi_i) = - x_{1i} + \phi (x_{2i}-1) + x_{3i}$, with $\phi=0$ for Scenario 1 and $\phi=1$ for Scenario 2, respectively.  The parameter of interest is $\theta = \mathbb{E}(Y)$. The population size is 
 $N = 1,000$. We considered three choices of the calibration variables in the proposed PS method. 

    \begin{enumerate}
        \item The proposed PS estimator using $(1, x_1, x_2)$ as the calibration variable
        \item The proposed PS estimator using $(1, x_1, x_3)$ as the calibration variable
        \item The proposed PS estimator using $(1, x_1, x_2, x_3)$ as the calibration variable
    \end{enumerate}

Monte Carlo samples of size $B=1,000$ are used to compute the Monte Carlo bias, Monte Carlo standard errors and the root mean squared errors of the estimators considered.   
The simulation result in Table 1 shows that the PS estimator using $(1, x_1, x_2)$ is the most efficient among the three calibration estimators considered. That is, using calibration variables in the outcome model only achieves the best efficiency.  In Scenario 2, the calibration estimator using the covariates for the outcome model is more efficient than the calibration estimator using the covariates for the selection model. 
The simulation result is consistent with our theory in Remark 3.

\begin{table}
\caption{Simulation summary of three calibration estimators for Simulation Study Two} 
\begin{center} 
\fbox{\begin{tabular}{r|rrr|rrr}
   & \multicolumn{3}{|c|}{Scenario 1} & \multicolumn{3}{|c}{Scenario 2} \\
   \cline{2-7}  
 Method & BIAS & SE & RMSE & BIAS & SE & RMSE\\ 
  \hline
  a & 0.00 & 0.063 & 0.063 &  0.00 & 0.072 & 0.072 \\ 
  b & 0.00 & 0.083 & 0.083 & -0.34 & 0.086 & 0.351\\ 
  c & 0.00 & 0.070 & 0.070 &  0.00 & 0.085 & 0.085 \\ 
\end{tabular}
}
\end{center} 
\end{table}

\section{Conclusion}
In handling the selection bias of the voluntary sample, the problem of estimating the inverse propensity score is recast as a problem of estimating the density ratio function. The density ratio function is defined for the balancing score function, and the propensity scores that arise can be efficient if the balancing score includes the actual mean function.
The information projection technique builds a self-efficient propensity score estimating function. The 
variable selection technique for the outcome model can be utilized to generate efficient propensity score weights.
The proposed method can be utilized as a unifying tool for merging information from multiple data sources.

There are several directions for further extensions of the proposed method. The proposed method can be extended to handle data integration problems \citep{chen2020doubly} 
 in survey sampling. 
 Also, the proposed method is based on the assumption of missing at random. Extension to nonignorable nonresponse \citep{kim2011semiparametric} can be also an interesting research direction. 
 Instead of using Kullback-Leibler divergence, we may use Hellinger divergence  \citep{etho2021, li2019b} or $\gamma$-power divergence \citep{eguchi2021} to achieve some robustness. 
 In addition, the proposed method can be used for causal inference, including the estimation of the average treatment effect from observational studies (\citealp{yangding2020}).  Developing tools for causal inference using the proposed  method will be an important extension of this research. 
 
\section*{Acknowledgement}

The authors would like to thank professors Wayne A. Fuller and Zhiqiang Tan for their constructive comments. 

\bibliographystyle{chicago}
\bibliography{ref_refine}

\end{document}